\documentclass[runningheads]{llncs}

\usepackage[T1]{fontenc}
\usepackage[utf8]{inputenc}

\usepackage{graphicx}
\usepackage{tikz}
\usepackage{tkz-graph}
\usepackage{tikz-qtree} 
\usetikzlibrary{calc}
\usetikzlibrary{positioning}
\usetikzlibrary{arrows.meta}
\usetikzlibrary{shapes}

\tikzset{SEQ/.style={draw,circle,append after command={
        	[shorten >=\pgflinewidth, shorten <=\pgflinewidth,]
        	(\tikzlastnode.north) edge (\tikzlastnode.south)
        	(\tikzlastnode.east) edge (\tikzlastnode.west)
        }
    }
}
\tikzset{XSEL/.style={draw,circle,append after command={
        	[shorten >=\pgflinewidth, shorten <=\pgflinewidth,]
        	node[circle,fill=black,scale=0.2] at (\tikzlastnode) (a) {}        
        }
    }
}
\tikzset{ISEL/.style={draw,circle,append after command={
        	[shorten >=\pgflinewidth, shorten <=\pgflinewidth,]
        	node[circle,draw,scale=0.2] at (\tikzlastnode) (a) {}        
        }
    }
}
\tikzset{PAR/.style={draw,circle,append after command={
			[shorten >=\pgflinewidth, shorten <=\pgflinewidth,]
			($(\tikzlastnode.north west)!0.5!(\tikzlastnode.north)$) edge ($(\tikzlastnode.south west)!0.5!(\tikzlastnode.south)$)
			($(\tikzlastnode.north)!0.5!(\tikzlastnode.north east)$) edge ($(\tikzlastnode.south)!0.5!(\tikzlastnode.south east)$)
        }
    }
}
\tikzset{AGG/.style={draw,circle,append after command={
        	[shorten >=\pgflinewidth, shorten <=\pgflinewidth,]
        	node[diamond,draw,scale=0.2] at (\tikzlastnode) (a) {}        
        }
    }
}

\newcommand{\SEQ}{\tikz{\node[SEQ,scale=0.6]{};}}

\newcommand{\PAR}{\tikz{\node[PAR,scale=0.6]{};}}
\newcommand{\AGG}{\tikz{\node[AGG,scale=0.6]{};}}

\usepackage{amssymb,mathtools,nccmath}
\usepackage[nounderscore]{syntax}

\newtheorem{notation}{Notation}

\usepackage{multirow}
\usepackage{array}

\begin{document}

\title{Towards an Algebra of Computon Spaces}
\author{Damian Arellanes\inst{1}}
\authorrunning{D. Arellanes}

\institute{Lancaster University, Lancaster, United Kingdom \\
\email{damian.arellanes@lancaster.ac.uk}}

\maketitle              

\begin{abstract}
Compositionality is a key property for dealing with complexity, which has been studied from many points of view in diverse fields. Particularly, the composition of individual computations (or programs) has been widely studied almost since the inception of computer science. Unlike existing composition theories, this paper presents an algebraic model not for composing individual programs but for inductively composing spaces of sequential and/or parallel constructs. We particularly describe the semantics of the proposed model and present an abstract example to demonstrate its application. 

\end{abstract}

\section{Introduction}

The Church-Turing thesis states that the intuitive notion of algorithms (or programs) is equivalent to that of a Turing Machine~\cite{sipser_introduction_2013}. The latter is an abstract device that receives an input, performs some computation and produces an output. Functions that can be computed by Turing Machines are called computable~\cite{turing_computable_1937}.

As the class of computable functions is closed under composition, the composition of two computable functions results in a higher-order computable function~\cite{sudkamp_languages_2005}. Equivalently, composing programs $p$ and $q$ results in a more complex program $r$ (known as composite)~\cite{arbab_reo_2004,achermann_calculus_2005,lau_introduction_2017,arellanes_evaluating_2020}. If $r$ defines a sequential composition $q \circ p$, then the computation of $p$ is followed by the computation of $q$. If $r$ defines a parallel composition $p || q$, then the computations of $p$ and $q$ are performed at the same time independently \cite{yanofsky_theoretical_2022}. In any case, the computation of $r$ halts if and only if the computations of $p$ and $q$ also halt \cite{cutland_computability_1980}. 

Rather than providing operators for composing individual programs, this paper presents an algebraic model for inductively composing spaces of computons through so-called composition operators. A computon is a function that defines sequencing or parallelisation, so the result of composition is a space of sequential and/or parallel constructs which can be further composed into higher-order spaces. Our model also provides an operator for reducing spaces which, like composition operators, satisfies totality.

Section~\ref{sec:model} presents the semantics of the proposed model which we refer to as Meronomic. Section~\ref{sec:example} uses Meronomic to describe an abstract example for the construction of higher-order spaces. Finally, Section~\ref{sec:conclusions} outlines the conclusions and describes future directions.

\section{The Meronomic Model}
\label{sec:model}

Meronomic is an algebraic model where computon spaces and composition operators are first-class entities. In this section, we describe its semantics.

\subsection{Semantics of Computon Spaces}
\label{sec:program-spaces-semantics}

A \emph{computon space} is a set of computons which are functions that define sequencing or parallelisation.\footnote{For the rest of the paper, the terms \emph{space} and \emph{computon space} are used interchangeably.} A primitive space is the simplest, indivisible unit of composition which contains a unique computon. Sequential, parallel or aggregated spaces are formed by the composition of multiple spaces and contain a potentially infinite number of computons. Below we present the semantics of computon spaces. For the semantics of composition, see Sections~\ref{sec:semantics-composition-operators}~and~\ref{sec:semantics-higherorder}.

\begin{notation}[Computon Space Universe]
Let $\mathbb{S}$ be the universe of computon spaces, $\mathbb{P}$ be the universe of primitive spaces and $\mathbb{C}$ be the universe of composite spaces such that $\mathbb{P} \subset \mathbb{S}$ and $\mathbb{C} \subset \mathbb{S}$. We denote $\emptyset$ as the \emph{empty space} which is the space with no computons at all.
\end{notation}

\begin{definition}[Primitive Computon]
A \emph{primitive computon} is a function $\{1\} \rightarrow H$ where $H$ is the set of all functions that can be computed by a halting Turing Machine. 
\end{definition} 

\begin{definition}[Primitive Space]
A \emph{primitive space} $S \in \mathbb{P}$ is a singleton set where $p \in S$ a primitive computon.
\end{definition}

\begin{definition}[Sequential Computon] \label{def:sequential-computon}
A \emph{sequential computon} $p$ is a partial function $\mathbb{Z}^+ \rightarrow A$ where $Dom(p) \neq \emptyset$, $A \subset \mathbb{S}$, $|A| \geq 2$ and $S \notin A$. As $p$ defines a strictly increasing (non-empty) sequence, denoted by $(p)_{i \in \mathbb{Z}^+}$,there exists a bijection ${f:[1,|Dom(p)|] \cap \mathbb{Z}^+ \rightarrow Dom(p)}$ given by ${f(x)=x}$. 
\end{definition}

\begin{definition}[Sequential Space]
A \emph{sequential space} $S \in \mathbb{C}$ is a set where each $p \in S$ is a sequential computon.
\end{definition}

\begin{definition}[Parallel Computon]
A \emph{parallel computon} $p$ is a (partial or total) function $A \rightarrow \mathbb{Z}^+$ where $A \subset \mathbb{S}$, $|A| \geq 2$ and $S \notin A$.
\end{definition}

\begin{definition}[Parallel Space]
A \emph{parallel space} $S \in \mathbb{C}$ is a set where each $p \in S$ is a parallel computon.
\end{definition}

\begin{definition}[Aggregated Space]
An \emph{aggregated space} $S \in \mathbb{C}$ is a set where each $p \in S$ is a sequential or a parallel computon.
\end{definition}

\begin{definition}[Computon Space Subsumption] \label{def:space-subsumption}
We say that ${S_2 \in \mathbb{S}}$ is subsumed by $S_1 \in \mathbb{S}$, written $S_2  \sqsubset S_1$, if there exists some computon $p \in S_1$ such that $S_2 \in Dom(p)$ or $S_2 \in Cod(p)$. 
\end{definition}

\begin{remark} \label{rem:computional-subsumption}
Given a computon space $S$, we have that $S \in \mathbb{P} \iff S \neq \emptyset~\land~{\nexists S_i \in \mathbb{S}}$ such that $S_i \sqsubset S$.
\end{remark}

\begin{remark} \label{rem:composite-subsumption}
Given a computon space $S$, we have that $S \in \mathbb{C} \iff {S \neq \emptyset}~\land$\\ ${\exists S_1,S_2 \in \mathbb{S}}$ such that ${S_1,S_2 \sqsubset S}$.
\end{remark}  

\subsection{Semantics of Composition Operators}
\label{sec:semantics-composition-operators}

A composite space is formed by the composition of two or more spaces in a hierarchical bottom-up manner. For this, Meronomic provides three composition operators: (i) sequencer, (ii) paralleliser and (iii) aggregator. As each of them returns a computon space, operators can be composed into more complex ones. Accordingly, Meronomic defines an algebra of computon spaces in which $\mathbb{S}$ is closed under sequencing, parallelisation and aggregation. 

\begin{definition}[Sequencer] \label{def:sequencer}
The sequencer operator $\SEQ$ is a function that takes $n \geq 2$ (non-empty) spaces and produces a sequential space:  
\begin{equation}
  \SEQ: \mathbb{S}^n \rightarrow \mathbb{C}
\end{equation}

More concretely, given a tuple $(S_1,S_2,\ldots,S_n) \in \mathbb{S}^n$, a sequential space $S \in \mathbb{C}$ is given by:\footnote{An n-tuple $(S_1,S_2,\ldots,S_n)$ of computon spaces can be defined as a surjective function ${\{1,2,\ldots,n\} \rightarrow \{S_1,S_2,\ldots,S_n\}}$ so that ${Im(S_1,S_2,\ldots,S_n)=\{S_1,S_2,\ldots,S_n\}}$.}

\begin{align*}
S &= \SEQ(S_1,S_2,\ldots,S_n) = Im(S_1,S_2,\ldots,S_n)^{\mathbb{Z}^+} 
\end{align*}
where each $p \in S$ is a sequential computon, i.e., a partial function from $\mathbb{Z}^+$ to $\{S_1,S_2,\ldots,S_n\}$. 
\end{definition}

\begin{example}
Given $S_1,S_2,S_3,S_4 \in \mathbb{S}$, we can construct the sequential space ${S=\SEQ(S_1,S_2,S_3,S_4) = \{S_1,S_2,S_3,S_4\}^{\mathbb{Z}^+}}$ which can diagramatically be represented as shown in Fig.~\ref{fig:example-sequential-space}.

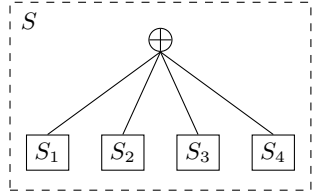
\begin{figure}[h]
\center
\begin{tikzpicture}[level distance=1.5cm,
  level 1/.style={sibling distance=1cm},
  level 2/.style={sibling distance=1cm}]
  \node(seq1)[SEQ,scale=1]{}
    child {node(s1)[draw] {$S_1$}}
    child {node(s2)[draw]{$S_2$}}
    child {node(s3)[draw]{$S_3$}}
    child {node(s4)[draw]{$S_4$}};

\node at ($(seq1)+(-1.75,0.25)$) (s) {$S$}; 
\draw[dashed]($(seq1) + (-2,0.5)$)rectangle($(s4) + (0.5,-0.5)$);
\end{tikzpicture}
\caption{Diagram of the sequential space ${S=\protect\SEQ(S_1,S_2,S_3,S_4)=\{S_1,S_2,S_3,S_4\}^{\mathbb{Z}^+}}$.}
\label{fig:example-sequential-space}
\end{figure} 

By Definition~\ref{def:sequencer}, $S$ contains all the sequential computons over four spaces. If we treat functions as sets, one of such computons is ${\{(1,S_2),(2,S_1),(3,S_4),(4,S_3)\}}$ which defines the sequence $\langle S_2,S_1,S_4,S_3 \rangle$. The diagrammatic interpretation of this computon is illustrated in Fig.~\ref{fig:example-sequential-program}.

\begin{figure}[h]
\center
\begin{tikzpicture}
\node[draw] at (0,0) (s2) {$S_2$};
\node[draw] at (1,0) (s1) {$S_1$};
\node[draw] at (2,0) (s4) {$S_4$};
\node[draw] at (3,0) (s3) {$S_3$};

\draw[->] (s2) -- (s1);
\draw[->] (s1) -- (s4);
\draw[->] (s4) -- (s3);
\end{tikzpicture}
\caption{Diagrammatic interpretation of the sequential computon $\langle S_2,S_1,S_4,S_3 \rangle$ which is a member of the sequential space $\protect\SEQ(S_1,S_2,S_3,S_4)$.}
\label{fig:example-sequential-program}
\end{figure}
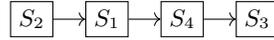
\end{example}

A sequential computon $p \in \SEQ(S_1,S_2,\ldots,S_n)$ can be \emph{bijective}, \emph{injective/non-surjective}, \emph{non-injective/surjective} or \emph{non-injective/non-surjective}. It is injective when $(p)_{i \in \mathbb{Z}^+}$ does not have any repeated spaces from $\{S_1,S_2,\ldots,S_n\}$, and it is surjective when each space from $\{S_1,S_2,\ldots,S_n\}$ appears at least once in $(p)_{i \in \mathbb{Z}^+}$. By combining these properties, we can build the other classes of sequential computons (see Table~\ref{tab:sequential-classes}).

\begin{table*}
\caption{Possible classes for a sequential computon $p: \mathbb{Z}^+ \rightarrow \{S_1,S_2,\ldots,S_n\}$.}
\begin{center}
\begin{tabular}{ |m{14em}|m{20em}| } 
\hline
\textbf{Class of $p$} & \textbf{Description} \\
\hline
Bijective & Each space in $\{S_1,S_2,\ldots,S_n\}$ appears only once in $(p)_{i \in \mathbb{Z}^+}$ \\
\hline
Injective/Non-Surjective & Some spaces in $\{S_1,S_2,\ldots,S_n\}$ appear only once in $(p)_{i \in \mathbb{Z}^+}$ \\
\hline
Non-injective/Surjective & All the spaces in $\{S_1,S_2,\ldots,S_n\}$ appear in $(p)_{i \in \mathbb{Z}^+}$ and at least one of them is repeated \\
\hline
Non-injective/Non-Surjective & Some spaces in $\{S_1,S_2,\ldots,S_n\}$ appear in $(p)_{i \in \mathbb{Z}^+}$ and at least one of them is repeated \\
\hline
\end{tabular}
\end{center}
\label{tab:sequential-classes}
\end{table*}

Moreover, we have that ${|(p)_{i \in \mathbb{Z}^+}|=n}$ when $p$ is bijective, $|(p)_{i \in \mathbb{Z}^+}| < n$ when $p$ is injective/non-surjective, ${|(p)_{i \in \mathbb{Z}^+}| > n}$ when $p$ is non-injective/surjective and $|(p)_{i \in \mathbb{Z}^+}| \geq n$ when $p$ is non-injective/non-surjective.

When $p$ is non-surjective, it is possible to have a sequence of length one or a sequence with at least one absent space. 

\begin{definition} 
Given a sequential computon ${p \in \SEQ(S_1,S_2,\ldots,S_n)}$, we say that ${S_i \in \{S_1,S_2,\ldots,S_n\}}$ is absent from $p$ if $S_i \notin Im(p)$. 
\end{definition}

\begin{proposition}
Any sequential space has the cardinality of the continuum.
\end{proposition}
\begin{proof}

Let $A=\{S_1,S_2,\ldots,S_n\}$ be a finite set of ${n \geq 2}$ computon spaces. By Definition~\ref{def:sequencer}, a sequential space $S$ over $A$ is $A^{\mathbb{Z}^+}$, i.e., the set of all strictly increasing (non-empty) sequences over $A$. Assuming the continuum hypothesis, we now prove that $|S|=|A^{\mathbb{Z}^+}|=\aleph_1$.

Considering that $|\mathbb{Z}^+|=\aleph_0$ and $|A| \geq 2$, we apply the multiplication principle of counting to get the following:

\begin{align}
|\{S_1,S_2\}^{\mathbb{Z}^+}|=2^{\aleph_0} & ~~\text{for}~|A|=2~\text{(i.e., the smallest}~A\text{)} \nonumber \\
|\{S_1,S_2,S_3\}^{\mathbb{Z}^+}|=3^{\aleph_0} & ~~\text{for}~|A|=3 \nonumber \\
\vdots & \nonumber \\
|\{S_1,S_2,S_3,\ldots,S_n\}^{\mathbb{Z}^+}|=n^{\aleph_0} & ~~\text{for}~|A|=n~\text{(i.e., the largest}~A\text{)} \nonumber
\end{align}

The chain of inclusions ${\{S_1,S_2\}^{\mathbb{Z}^+} \subset \{S_1,S_2,S_3\}^{\mathbb{Z}^+} \subset \cdots \subset \{S_1,S_2,S_3,\ldots,S_n\}^{\mathbb{Z}^+}}$ implies that $2^{\aleph_0} \leq 3^{\aleph_0} \leq \cdots \leq n^{\aleph_0}$. To prove that $2^{\aleph_0} = 3^{\aleph_0} = \cdots = n^{\aleph_0}$, we show that there exists a bijection $f:\{S_1,S_2,S_3,\ldots,S_n\}^{\mathbb{Z}^+} \rightarrow \{S_1,S_2\}^{\mathbb{Z}^+}$.

Let $p$ be some sequential computon in $\{S_1,S_2,S_3,\ldots,S_n\}^{\mathbb{Z}^+}$ and $q$ some sequential computon in $\{S_1,S_2\}^{\mathbb{Z}^+}$. A bijection $f$ can be given as follows: for all indices $i \in \mathbb{Z}^+$, $q(i)=S_1$ if $p(i)=S_j$ for some positive odd integer $j$. Conversely, $q(i)=S_2$ if $p(i)=S_k$ for some positive even integer $k$.

Consequently, $|S|=n^{\aleph_0}=2^{\aleph_0}$ for $|A| \geq 2$. By the continuum hypothesis, $2^{\aleph_0}=\aleph_1 \implies |S|=\aleph_1$ for any finite set of at least two spaces.
\end{proof}

\begin{definition}[Paralleliser] \label{def:paralleliser}
The paralleliser operator $\PAR$ is a function that takes $n \geq 2$ (non-empty) spaces and produces a parallel space:  
\begin{equation}
  \PAR: \mathbb{S}^n \rightarrow \mathbb{C}
\end{equation}

More concretely, $S \in \mathbb{C}$ is constructed from a tuple $(S_1,S_2,\ldots,S_n) \in \mathbb{S}^n$ as follows:

\begin{align*}
S &= \PAR(S_1,S_2,\ldots,S_n) = {\mathbb{Z}^+}^{Im(S_1,S_2,\ldots,S_n)}
\end{align*}
where each parallel computon $p \in S$ is a (partial or total) non-surjective function that maps each space in $\{S_1,S_2,\ldots,S_n\}$ to an integer representing the number of parallel instances. 
\end{definition}

\begin{example}
Given $S_1,S_2,S_3,S_4 \in \mathbb{S}$, we can construct the parallel space $S={\PAR(S_1,S_2,S_3,S_4) = {\mathbb{Z}^+}^{\{S_1,S_2,S_3,S_4\}}}$ which can diagramatically be represented as shown in Fig.~\ref{fig:example-parallel-space}.

\begin{figure}[h]
\center
\begin{tikzpicture}[level distance=1.5cm,
  level 1/.style={sibling distance=1cm},
  level 2/.style={sibling distance=1cm}]
  \node(par1)[PAR,scale=1]{}
    child {node(s1)[draw] {$S_1$}}
    child {node(s2)[draw]{$S_2$}}
    child {node(s3)[draw]{$S_3$}}
    child {node(s4)[draw]{$S_4$}};

\node at ($(par1)+(-1.75,0.25)$) (s) {$S$}; 
\draw[dashed]($(par1) + (-2,0.5)$)rectangle($(s4) + (0.5,-0.5)$);
\end{tikzpicture}
\caption{Diagram of the parallel space ${S=\protect\PAR(S_1,S_2,S_3,S_4)={\mathbb{Z}^+}^{\{S_1,S_2,S_3,S_4\}}}$.}
\label{fig:example-parallel-space}
\end{figure}
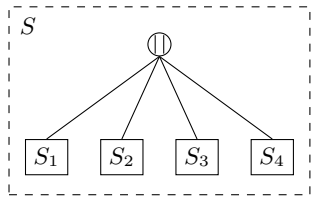 

Treating functions as sets, a parallel computon $p \in S$ can be \\${\{(S_1,2),(S_2,1),(S_3,3),(S_4,1)\}}$ whose interpretation is illustrated in Fig.~\ref{fig:example-parallel-program}.\footnote{For simplicity, for the rest of the paper we do not show the labels for the fork and join constructs. We assume that they are clear from the context.} 

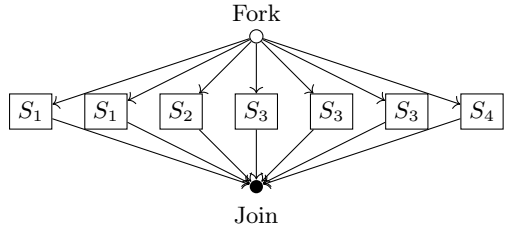
\begin{figure}[h]
\center
\begin{tikzpicture}

\node[circle,draw=black,inner sep=0pt,minimum size=5pt,label=Fork] at (3,2) (fork) {};
\node[draw] at (0,1) (i1) {$S_1$};
\node[draw] at (1,1) (i2) {$S_1$};
\node[draw] at (2,1) (i3) {$S_2$};
\node[draw] at (3,1) (i4) {$S_3$};
\node[draw] at (4,1) (i5) {$S_3$};
\node[draw] at (5,1) (i6) {$S_3$};
\node[draw] at (6,1) (i7) {$S_4$};
\node[circle,fill=black,inner sep=0pt,minimum size=5pt,label={[yshift=-0.7cm]Join}] at (3,0) (join) {};

\draw[->] (fork) -- (i1); \draw[->] (i1) -- (join);
\draw[->] (fork) -- (i2); \draw[->] (i2) -- (join);
\draw[->] (fork) -- (i3); \draw[->] (i3) -- (join);
\draw[->] (fork) -- (i4); \draw[->] (i4) -- (join);
\draw[->] (fork) -- (i5); \draw[->] (i5) -- (join);
\draw[->] (fork) -- (i6); \draw[->] (i6) -- (join);
\draw[->] (fork) -- (i7); \draw[->] (i7) -- (join);
\end{tikzpicture}
\caption{Diagrammatic interpretation of the parallel computon $\{(S_1,2),(S_2,1),(S_3,3),(S_4,1)\}$ which is a member of the space $\protect\PAR(S_1,S_2,S_3,S_4)$.}
\label{fig:example-parallel-program}
\end{figure}
\end{example}

Any parallel computon $p \in \PAR(S_1,S_2,\ldots,S_n)$ parallelises some or all the spaces from $\{S_1,S_2,\ldots,S_n\}$. The former case occurs when $p$ is a \emph{partial function} whereas the latter happens when $p$ is a \emph{total function}. As $p$ is always a function, we have that at least one space from $\{S_1,S_2,\ldots,S_n\}$ is parallelised.

Similarly, if $p$ is \emph{injective} all the parallelised spaces have a different number of parallel instances. Otherwise, at least two spaces have the same number of parallel instances. The classes of parallel computons are described in Table~\ref{tab:parallel-classes}, viz. \emph{partial/injective}, \emph{partial/non-injective}, \emph{total/injective} and \emph{total/non-injective}.

\begin{table}[!h]
\caption{Possible classes for a parallel computon $p:\{S_1,S_2,\ldots,S_n\} \rightarrow \mathbb{Z}^+$.}
\begin{center}
\begin{tabular}{ |m{10em}|m{22em}| } 
\hline
\textbf{Class of $p$} & \textbf{Description} \\
\hline
Partial/Injective & Some spaces in $\{S_1,S_2,\ldots,S_n\}$ are parallelised, each with a different number of parallel instances \\
\hline
Total/Injective & All the spaces in $\{S_1,S_2,\ldots,S_n\}$ are parallelised, each with a different number of parallel instances \\
\hline
Partial/Non-Injective & Some spaces in $\{S_1,S_2,\ldots,S_n\}$ are parallelised and at least two of them have the same number of parallel instances \\
\hline
Total/Non-Injective & All the spaces in $\{S_1,S_2,\ldots,S_n\}$ are parallelised and at least two of them have the same number of parallel instances \\
\hline
\end{tabular}
\end{center}
\label{tab:parallel-classes}
\end{table}

It is important to note that, when $p$ is partial, it is possible to have a parallel computon parallelising only one space or that has at least one absent space. 

\begin{definition} 
Given a parallel computon ${p \in \PAR(S_1,S_2,\ldots,S_n)}$, we say that \\${S_i \in \{S_1,S_2,\ldots,S_n\}}$ is absent from $p$ if $S_i \notin Dom(p)$. 
\end{definition}

\begin{proposition}
Every parallel space is countably infinite.
\end{proposition}
\begin{proof}
Let $A=\{S_1,S_2\ldots,S_n\}$ be a finite set of ${n \geq 2}$ spaces. By Definition~\ref{def:paralleliser}, a parallel space $S$ over $A$ is ${\mathbb{Z}^+}^A$, i.e., the set of all functions $A \rightarrow \mathbb{Z}^+$. We now prove that $|S|=|{\mathbb{Z}^+}^A|=\aleph_0$.

Considering that $|\mathbb{Z}^+|=\aleph_0$ and $n \geq 2$, by the multiplication principle of counting we have that ${|S|=|{\mathbb{Z}^+}^A|={\aleph_0}^n}$. As ${\aleph_0}^n=\aleph_0$ for $n > 0$, we have that ${|S|=|{\mathbb{Z}^+}^A|={\aleph_0}^n=\aleph_0}$ for any finite set of at least two spaces.
\end{proof}

\begin{definition}[Aggregator] \label{def:aggregator}
The aggregator operator $\AGG$ is a function that takes $n \geq 2$ spaces and produces an aggregated space:  
\begin{equation}
  \AGG: \mathbb{S}^n \rightarrow \mathbb{C}
\end{equation}

More concretely, given a tuple $(S_1,S_2,\ldots,S_n) \in \mathbb{S}^n$, an aggregated space $S \in \mathbb{C}$ is given as follows:

\begin{align*}
S &= \AGG(S_1,S_2,\ldots,S_n) = \bigcup\limits_{1 \leq i \leq n} S_i
\end{align*}
where $S_i \in Im(S_1,S_2,\ldots,S_n)$.
\end{definition}

\begin{example}
Fig.~\ref{fig:example-aggregated-space} shows the diagrammatic representation of the aggregated space $S={\AGG(S_1,S_2,S_3,S_4) = S_1 \cup S_2 \cup S_3 \cup S_4}$ where $S \in \mathbb{C}$ and ${S_1,S_2,S_3,S_4 \in \mathbb{S}}$.

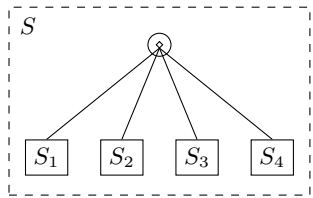
\begin{figure}[h]
\center
\begin{tikzpicture}[level distance=1.5cm,
  level 1/.style={sibling distance=1cm},
  level 2/.style={sibling distance=1cm}]
  \node(agg1)[AGG,scale=1]{}
    child {node(s1)[draw] {$S_1$}}
    child {node(s2)[draw]{$S_2$}}
    child {node(s3)[draw]{$S_3$}}
    child {node(s4)[draw]{$S_4$}};

\node at ($(agg1)+(-1.75,0.25)$) (s) {$S$}; 
\draw[dashed]($(agg1) + (-2,0.5)$)rectangle($(s4) + (0.5,-0.5)$);
\end{tikzpicture}
\caption{Diagram of the aggregated space ${S=\protect\AGG(S_1,S_2,S_3,S_4)=S_1 \cup S_2 \cup S_3 \cup S_4}$.}
\label{fig:example-aggregated-space}
\end{figure} 

Unlike the sequencer and the paralleliser operator, the aggregator does not define a space of new computons, but just an aggregation of existing ones. For instance, in our example, all the computons in $S_i$ are also in $S$ for all $1 \leq i \leq 4$.
\end{example}

\subsection{Semantics of Higher-Order Composition}
\label{sec:semantics-higherorder}

Computon spaces are compositional because every composition operator returns a space that can be further composed inductively (i.e., in a hierarchical bottom-up manner). Thus, compositions are multi-level hierarchical structures in which primitive spaces always lie at the bottom. 

\begin{notation} [Computon Space Order]
The order of a space $S$, denoted by $\Theta(S)$, defines the level of $S$ in a hierarchical composition structure. 
\end{notation}

\begin{remark}
Every primitive space is a $0$-order space: $\Theta(S)=0 \iff S \in \mathbb{P}$. 
\end{remark}

\begin{definition}[First-Order Space]
We say that a computon space $S$ is a first-order space if all the spaces it subsumes are primitive: $\Theta(S)=1 \iff {S_1,S_2,\ldots,S_n \sqsubset S}$ and $S_i \in \mathbb{P}$ for all $1 \leq i \leq n$ with $n \geq 2$.
\end{definition}

\begin{definition}[K-Order Space]
For $k \geq 1$, a computon space $S$ is a $k$-order space if it subsumes at least one space of $k-1$ order and does not subsume spaces of order equal or greater than $k$. More precisely, we have that ${\Theta(S)=k} \iff {\exists S_i \in \mathbb{S}, S_i \sqsubset S~\land~\Theta(S_i)=k-1}$ and ${\nexists S_j \in \mathbb{S},} {S_j \sqsubset S}~\land~\Theta(S_j) \geq k$.
\end{definition}

\begin{example} [Second-Order Space]
Let $S_1,S_2,S_3,S_4,S_5 \in \mathbb{P}$ and $S_6 \in \mathbb{C}$ be the first-order sequential space resulting from the operation $\SEQ(S_1,S_2)$. Diagramatically, $S_6$ can be represented as shown in Fig.~\ref{fig:sequential-s6}.

\begin{figure}[!h]
\center
\begin{tikzpicture}[level distance=1.5cm,
  level 1/.style={sibling distance=1cm},
  level 2/.style={sibling distance=1cm}]
  \node(seq1)[SEQ,scale=1]{}
    child {node(s1)[draw] {$S_1$}}
    child {node(s2)[draw]{$S_2$}};

\node at ($(seq1)+(-0.75,0.25)$) (s6) {$S_6$}; 
\draw[dashed]($(seq1) + (-1,0.5)$)rectangle($(s2) + (0.5,-0.5)$);
\end{tikzpicture}
\caption{Diagram of the first-order sequential space ${S_6=\protect\SEQ(S_1,S_2)}$.}
\label{fig:sequential-s6}
\end{figure}
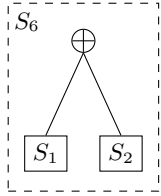

Likewise, we can construct the first-order aggregated space $S_7$ through the operation $\AGG(S_3,S_4)$ to yield the diagram depicted in Fig.~\ref{fig:aggregated-s7}.

\begin{figure}[!h]
\center
\begin{tikzpicture}[level distance=1.5cm,
  level 1/.style={sibling distance=1cm},
  level 2/.style={sibling distance=1cm}]
  \node(agg1)[AGG,scale=1]{}
    child {node(s3)[draw] {$S_3$}}
    child {node(s4)[draw]{$S_4$}};

\node at ($(agg1)+(-0.75,0.25)$) (s7) {$S_7$}; 
\draw[dashed]($(agg1) + (-1,0.5)$)rectangle($(s4) + (0.5,-0.5)$);
\end{tikzpicture}
\caption{Diagram of the first-order aggregated space ${S_7=\protect\AGG(S_3,S_4)}$.}
\label{fig:aggregated-s7}
\end{figure}
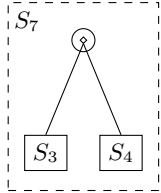

To further compose $S_5$, $S_6$ and $S_7$ into the second-order parallel space $S_8$, we apply the operation $\PAR(S_5,S_6,S_7)$. The diagram of $S_8$ is shown in Fig.~\ref{fig:parallel-s8}. 

\begin{figure}[!h]
\center
\begin{tikzpicture}[level distance=1.5cm,
  level 1/.style={sibling distance=1cm},
  level 2/.style={sibling distance=1cm}]
  \node(par1)[PAR,scale=1]{}
    child {node(s5)[draw] {$S_5$}}
    child {node(s6)[draw,dashed]{$S_6$}}
    child {node(s7)[draw,dashed]{$S_7$}};

\node at ($(par1)+(-1.25,0.25)$) (s8) {$S_8$};
\draw[dashed]($(par1) + (-1.5,0.5)$)rectangle($(s7) + (0.5,-0.5)$);
\end{tikzpicture}
\caption{Diagram of the second-order parallel space ${S_8=\protect\PAR(S_5,S_6,S_7)}$ where $S_6$ and $S_7$ are first-order spaces (see Figs.~\ref{fig:sequential-s6}~and~\ref{fig:aggregated-s7}, respectively).}
\label{fig:parallel-s8}
\end{figure}
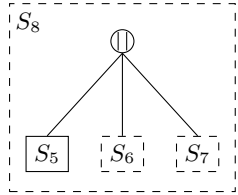

If we expand the composites $S_6$ and $S_7$, we can explicitly observe that $S_8$ is the second-order composite $\PAR(S_5,\SEQ(S_1,S_2),\AGG(S_3,S_4))$ whose complete structure is illustrated in Fig.~\ref{fig:parallel-s8-expanded}. More precisely:

\begin{figure}[!h]
\center
\begin{tikzpicture}[level distance=1.5cm,
  level 1/.style={sibling distance=2.5cm},
  level 2/.style={sibling distance=1cm}]
  \node(par1)[PAR]{}
	child {node(s5)[draw]{$S_5$}
    }        
    child {node(seq1)[SEQ,scale=1] {}
      child {node(s1)[draw]{$S_1$}}
      child {node(s2)[draw]{$S_2$}}
    }
    child {node(agg1)[AGG,scale=1] {}
      child {node(s3)[draw]{$S_3$}}
      child {node(s4)[draw]{$S_4$}}
    };

\node at ($(seq1)+(-0.75,0.25)$) (s6) {$S_6$};
\draw[dashed]($(seq1) + (-1,0.5)$)rectangle($(s2) + (0.5,-0.5)$);

\node at ($(agg1)+(-0.75,0.25)$) (s7) {$S_7$};
\draw[dashed]($(agg1) + (-1,0.5)$)rectangle($(s4) + (0.5,-0.5)$);

\node at ($(par1)+(-2.75,0.25)$) (s8) {$S_8$};
\draw[dashed]($(par1) + (-3,0.5)$)rectangle($(s4) + (0.75,-0.75)$);
\end{tikzpicture}
\caption{Expanded diagram of the second-order parallel space $S_8=\protect\PAR(S_5,S_6,S_7)$.}
\label{fig:parallel-s8-expanded}
\end{figure}
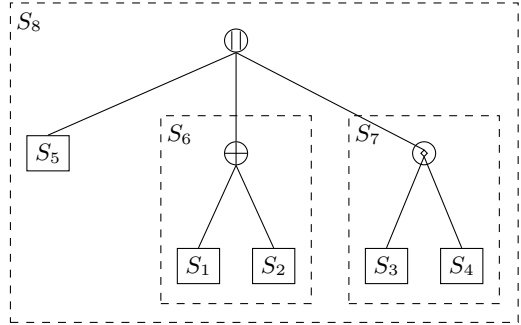

\begin{equation}
S_8=\PAR(S_5,S_6,S_7)=\PAR(S_5,\SEQ(S_1,S_2),\AGG(S_3,S_4))
\end{equation}

\end{example}

\subsection{Properties of Composition Operators}

In the Meronomic model, all composition operators satisfy totality (i.e., the closure axiom) since $\mathbb{C} \subset \mathbb{S}$. We now show that the sequencer and the paralleliser operators are commutative only, whereas the universe $\mathbb{S}$ under aggregation forms a commutative multiary monoid $(\mathbb{S},\AGG)$. Table~\ref{tab:operators-properties} summarises the supported properties by the proposed composition operators.

\begin{table*} 
\caption{Properties of Composition Operators.}
\begin{center}
\begin{tabular}{ |m{6em}|m{5em}|m{10em}|m{15.3em}| }  
\hline
 & \textbf{Identity} & \textbf{Commutativity} & \textbf{Associativity} \\
\hline
\textbf{Sequencer} & No & $\SEQ(S_1,S_2) = \SEQ(S_2,S_1)$ & No \\
\hline
\textbf{Paralleliser} & No & $\PAR(S_1,S_2) = \PAR(S_2,S_1)$ & No \\
\hline
\textbf{Aggregator} & $\AGG(S,\emptyset) = S$ & $\AGG(S_1,S_2) = \AGG(S_2,S_1)$ & $\AGG(S_1,\AGG(S_2,S_3))=\AGG(\AGG(S_1,S_2),S_3)$ \\
\hline
\end{tabular}
\end{center}
\label{tab:operators-properties}
\end{table*}

\subsubsection{Identity.}

The aggregator operator is the only one that satisfies identity. 

\begin{proposition}
For any space $S \in \mathbb{S}$, we have that $\AGG(S,\emptyset) = S$.
\end{proposition}
\begin{proof}
The union of any set with the empty set is the set itself.
\end{proof}

\subsubsection{Commutativity.}

This property is satisfied by all the composition operators, which entails that the order of space operands is unimportant when constructing higher-order spaces.

\begin{proposition}
Space sequencing is commutative. 
\end{proposition}
\begin{proof}
Let $S_1,S_2 \in \mathbb{S}$. By Definition~\ref{def:sequencer}, $p \in \SEQ(S_1,S_2) \iff p \in Im(S_1,S_2)^{\mathbb{Z}^+} \iff p \in \{S_1,S_2\}^{\mathbb{Z}^+}$. Likewise, ${p \in \SEQ(S_2,S_1)} \iff p \in Im(S_2,S_1)^{\mathbb{Z}^+} \iff {p \in \{S_1,S_2\}^{\mathbb{Z}^+}}$. Thus, $\SEQ(S_1,S_2) = \{S_1,S_2\}^{\mathbb{Z}^+} = \SEQ(S_2,S_1)$.
\end{proof}

\begin{proposition}
Space parallelisation is commutative. 
\end{proposition}
\begin{proof}
Let $S_1,S_2 \in \mathbb{S}$. By Definition~\ref{def:paralleliser}, $p \in \PAR(S_1,S_2) \iff p \in {\mathbb{Z}^+}^{Im(S_1,S_2)} \iff p \in {\mathbb{Z}^+}^{\{S_1,S_2\}}$. Likewise, $p \in \PAR(S_2,S_1) \iff p \in {\mathbb{Z}^+}^{Im(S_2,S_1)} \iff {p \in {\mathbb{Z}^+}^{\{S_1,S_2\}}}$. 
Thus, $\PAR(S_1,S_2) = {\mathbb{Z}^+}^{\{S_1,S_2\}} = \PAR(S_2,S_1)$.
\end{proof}

\begin{proposition}
Space aggregation is commutative. 
\end{proposition}
\begin{proof}
The proof follows from the fact that set union is commutative.
\end{proof}

\subsubsection{Associativity.}

Associativity allows rearranging the order of space operands when performing a composition operation. This property is only satisfied by the aggregator operator. 

\begin{proposition}
Space sequencing is not associative. 
\end{proposition}
\begin{proof}
Given $S_1,S_2,S_3 \in \mathbb{S}$, a counterexample is provided by the sequential computon $\langle S_1,S_1 \rangle$ which is in $\SEQ(S_1,\SEQ(S_2,S_3))$ but not in $\SEQ(\SEQ(S_1,S_2),S_3)$. Hence, $\SEQ(S_1,\SEQ(S_2,S_3))\neq\SEQ(\SEQ(S_1,S_2),S_3)$.
\end{proof}

\begin{proposition}
Space parallelisation is not associative. 
\end{proposition}
\begin{proof}
Given $S_1,S_2,S_3 \in \mathbb{S}$, a counterexample is provided by the parallel computon $\{(S1,5)\}$ which is in $\PAR(S_1,\PAR(S_2,S_3))$ but not in $\PAR(\PAR(S_1,S_2),S_3)$. Hence, $\PAR(S_1,\PAR(S_2,S_3))\neq\PAR(\PAR(S_1,S_2),S_3)$.
\end{proof}

\begin{proposition}
Space aggregation is associative. 
\end{proposition}
\begin{proof}
The proof follows from the fact that set union is associative.
\end{proof}

\subsection{Semantics of Computon Space Reduction}

To choose one or multiple computons from a space, the Meronomic model provides the \emph{reductor operator} which takes a space $S$ and produces a subset of $S$'s computons satisfying some condition. When this happens, we say that $S$ is reduced into $S'$.

\begin{definition}[Space Reductor] \label{def:selector}
The space reductor ${\sigma_{\phi}:\mathbb{S} \rightarrow \mathbb{S}}$ is a unary operator given by: 
\begin{align*}
\sigma_{\phi}(S) = \{p \mid p \in S~\land~\phi(p)\}
\end{align*}
where $S \in \mathbb{S}$, $\phi$ is a propositional formula and $\sigma_{\phi}(S) \subseteq S$.
\end{definition}

A propositional formula $\phi$ consists of terms connected by the logical operators $\land$, $\lor$ and $\lnot$. For a reduction $\sigma_{\phi}(S)$, each term of $\phi(p)$ is defined according to the nature of $S$. If $S$ is a sequential space, then each term can be ${p(i)=S_j}$ or $p(i)=p(k)$ or $|p|=l$ where $p$ is the free variable, $S_j \sqsubset S$ and $i,k,l \in \mathbb{Z}^+$. If $S$ is a parallel space, each term is $p(S_j)=i$ or $p(S_j)=p(S_k)$ or $S_j \notin Dom(p)$ where $p$ is the free variable, $i,j,k \in \mathbb{Z}^+$ and $S_j,S_k \sqsubset S$. When $S$ is an aggregated space, each term can be of the form $p \in S_i$ where $S_i \sqsubset S$ and $p$ is the free variable. A primitive space $S$ can only be reduced to itself via the superflous term $p \in S$. For conciseness, in this paper we do not reduce primitive spaces.

\begin{example}
Given a sequential space ${S=\SEQ(S_1,S_2,S_3,S_4,S_5)}$, we define the operation $\sigma_{p(1)=S_3~\land~p(2)=p(4)}(S)$ to select all the computons in $S$ where the first element in the sequence is $S_3$, and the second and fourth elements are equal. In this case, the resulting space will contain $\langle S_3,S_2,S_1,S_2 \rangle$ and $\langle S_3,S_3,S_1,S_3,S_1 \rangle$, among an infinite number of sequential computons. Should we prefer to reduce $S$ into the singleton space $S'=\{\langle S_3,S_2,S_1,S_2 \rangle\}$, we define the operation $\sigma_{p(1)=S_3~\land~p(2)=S_2~\land~p(3)=S_1~\land~p(4)=S_2~\land~|p|=4}(S)$. Diagramatically, reducing $S$ into $S'$ can be expressed as shown in Fig.~\ref{fig:sequential-space-reduction}.

\begin{figure}[h]
\center
\begin{tikzpicture}[level distance=1cm,
  level 1/.style={sibling distance=1cm},
  level 2/.style={sibling distance=1cm}]
  \node(seq1)[SEQ,scale=1]{}
    child {node(s1)[draw] {$S_1$}}
	child {node(s2)[draw] {$S_2$}}
	child {node(s3)[draw] {$S_3$}}
	child {node(s4)[draw] {$S_4$}}       
    child {node(s5)[draw]{$S_5$}};

\node at ($(seq1)+(-2.25,0.25)$) (s) {$S$}; 
\draw[dashed]($(seq1) + (-2.5,0.5)$)rectangle($(s5) + (0.5,-0.5)$);

\node[draw] at (-2.25,-3.5) (0) {$S_3$};
\node[draw] at (-0.75,-3.5) (1) {$S_2$};
\node[draw] at (0.75,-3.5) (2) {$S_1$};
\node[draw] at (2.25,-3.5) (3) {$S_2$};

\node at (0.5,-2) (sel) {$S'=\sigma_{p(1)=S_3~\land~p(2)=S_2~\land~p(3)=S_1~\land~p(4)=S_2~\land~|p|=4}(S)$}; 
\draw[-{Triangle[width=10pt,length=8pt]}, line width=3pt](-4,-1.5) -- (-4, -2.5);

\node at ($(1)+(-1.75,0.7)$) (s) {$S'$}; 
\draw[dashed]($(1) + (-2,0.9)$)rectangle($(3) + (0.5,-0.5)$);

\draw[->] (0) -- (1);
\draw[->] (1) -- (2);
\draw[->] (2) -- (3);
\end{tikzpicture}
\caption{Reducing ${S=\protect\SEQ(S_1,S_2,S_3,S_4,S_5)}$ into ${S'=\{\langle S_3,S_2,S_1,S_2 \rangle\}}$.}
\label{fig:sequential-space-reduction}
\end{figure}
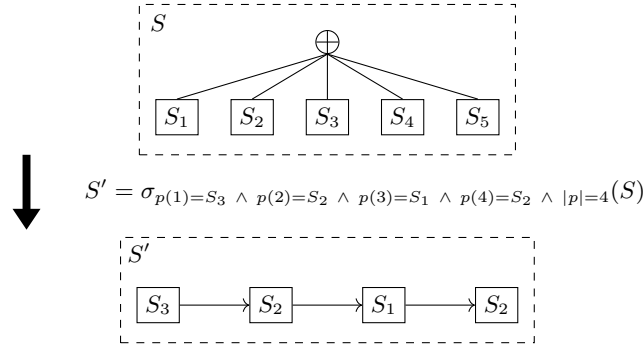
\end{example}

\begin{example}
Given a parallel computon space ${S=\PAR(S_1,S_2,S_3,S_4,S_5)}$, we define $\sigma_{p(S_3)=1~\land~p(S_2)=p(S_4)}(S)$ to select all the parallel computons in $S$ where $S_3$ has one parallel instance and $S_2$ and $S_4$ have the same number of parallel instances. In this case, the resulting space will contain $\{(S_1,7),(S_2,3),(S_3,1),(S_4,3)\}$ and $\{(S_2,2),(S_3,1),(S_4,2)\}$, among an infinite number of parallel computons. Should we prefer to reduce $S$ into the singleton space ${S'=\{\{(S_1,1),(S_2,2),(S_3,1)\}\}}$, we define the operation $\sigma_{p(S_1)=1~\land~p(S_2)=2~\land~p(S_3)=1~\land~S_4 \notin Dom(p)~\land~S_5 \notin Dom(p)}(S)$. Reducing $S$ into $S'$ can diagramatically be represented as shown in Fig.~\ref{fig:reduction-par}.

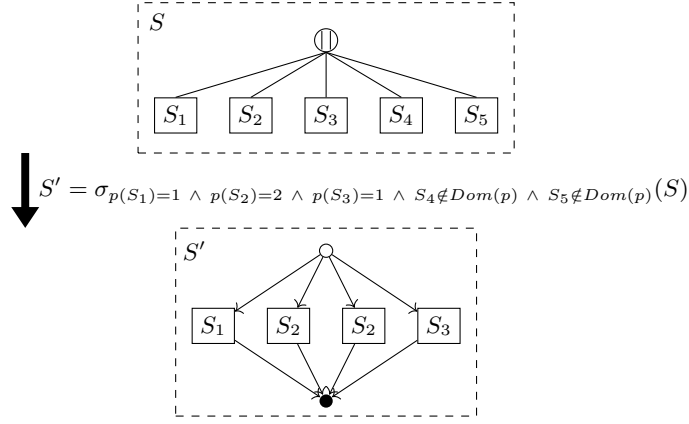
\begin{figure}[!h]
\center
\begin{tikzpicture}[level distance=1cm,
  level 1/.style={sibling distance=1cm},
  level 2/.style={sibling distance=1cm}]
  \node(seq1)[PAR,scale=1]{}
    child {node(s1)[draw] {$S_1$}}
	child {node(s2)[draw] {$S_2$}}
	child {node(s3)[draw] {$S_3$}}
	child {node(s4)[draw] {$S_4$}}       
    child {node(s5)[draw]{$S_5$}};

\node at ($(seq1)+(-2.25,0.25)$) (s) {$S$}; 
\draw[dashed]($(seq1) + (-2.5,0.5)$)rectangle($(s5) + (0.5,-0.5)$);

\node at (0.5,-2) (sel) {$S'=\sigma_{p(S_1)=1~\land~p(S_2)=2~\land~p(S_3)=1~\land~S_4 \notin Dom(p)~\land~S_5 \notin Dom(p)}(S)$}; 
\draw[-{Triangle[width=10pt,length=8pt]}, line width=3pt](-4,-1.5) -- (-4, -2.5);

\node[circle,draw=black,inner sep=0pt,minimum size=5pt] at (0,-2.8) (fork) {};
\node[draw] at (-1.5,-3.8) (i1) {$S_1$};
\node[draw] at (-0.5,-3.8) (i2) {$S_2$};
\node[draw] at (0.5,-3.8) (i3) {$S_2$};
\node[draw] at (1.5,-3.8) (i4) {$S_3$};
\node[circle,fill=black,inner sep=0pt,minimum size=5pt] at (0,-4.8) (join) {};

\node at ($(fork)+(-1.75,0)$) (s) {$S'$}; 
\draw[dashed]($(i1) + (-0.5,1.3)$)rectangle($(i3) + (1.5,-1.2)$);

\draw[->] (fork) -- (i1); \draw[->] (i1) -- (join);
\draw[->] (fork) -- (i2); \draw[->] (i2) -- (join);
\draw[->] (fork) -- (i3); \draw[->] (i3) -- (join);
\draw[->] (fork) -- (i4); \draw[->] (i4) -- (join);

\end{tikzpicture}
\caption{Reducing ${S=\protect\PAR(S_1,S_2,S_3,S_4,S_5)}$ into ${S'=\{\{(S_1,1),(S_2,2),(S_3,1)\}\}}$.}
\label{fig:reduction-par}
\end{figure}
\end{example}

\begin{example}
Given an aggregated space ${S=\AGG(S_1,S_2,S_3)}$, we reduce it into $\AGG(S_1,S_3)$ via the operation $\sigma_{p \in S_1~\land~p \in S_3}(S)$. Likewise, we can reduce $S$ into $S_2$ through the operation $\sigma_{p \in S_2}(S)$, as illustrated in Fig.~\ref{fig:reduction-agg}.

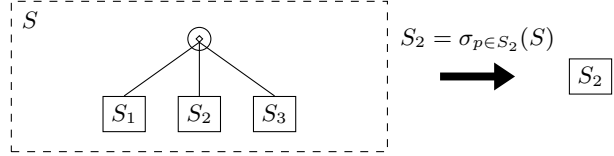
\begin{figure}[!h]
\center
\begin{tikzpicture}[level distance=1cm,
  level 1/.style={sibling distance=1cm},
  level 2/.style={sibling distance=1cm}]
  \node(agg1)[AGG,scale=1]{}
    child {node(s1)[draw] {$S_1$}}
	child {node(s2)[draw] {$S_2$}}
	child {node(s3)[draw] {$S_3$}};

\node at ($(agg1)+(-2.25,0.25)$) (s) {$S$}; 
\draw[dashed]($(agg1) + (-2.5,0.5)$)rectangle($(s5) + (0.5,-0.5)$);

\node[draw] at (5.2,-0.5) (0) {$S_2$};

\node at (3.7,0) (sel) {$S_2=\sigma_{p \in S_2}(S)$}; 
\draw[-{Triangle[width=10pt,length=8pt]}, line width=3pt](3.2,-0.5) -- (4.2, -0.5);
\end{tikzpicture}
\caption{Reducing ${S=\protect\AGG(S_1,S_2,S_3)}$ into $S_2$.}
\label{fig:reduction-agg}
\end{figure}
\end{example}

\begin{proposition}\label{prop:reductor-property-1}
Given $S \in \mathbb{S}$ and the propositional formulas $\phi$ and $\psi$, we have that ${\sigma_{\phi\land\psi}(S)=\sigma_{\phi}(S) \cap \sigma_{\psi}(S)}$.
\end{proposition}
\begin{proof}
By definition~\ref{def:selector}, $p \in \sigma_{\phi\land\psi}(S) \iff \phi(p)~\land~\psi(p) \iff p \in \sigma_{\phi}(S)~\land~p \in \sigma_{\psi}(S) \iff p \in \sigma_{\phi}(S) \cap \sigma_{\psi}(S)$.
\end{proof}

\begin{proposition}\label{prop:reductor-property-2}
Given $S \in \mathbb{S}$ and the propositional formulas $\phi$ and $\psi$, we have that ${\sigma_{\phi\lor\psi}(S)=\sigma_{\phi}(S) \cup \sigma_{\psi}(S)}$.
\end{proposition}
\begin{proof}
By definition~\ref{def:selector}, $p \in \sigma_{\phi\lor\psi}(S) \iff \phi(p)~\lor~\psi(p) \iff p \in \sigma_{\phi}(S)~\lor~p \in \sigma_{\psi}(S) \iff p \in \sigma_{\phi}(S) \cup \sigma_{\psi}(S)$.
\end{proof}

\begin{proposition}\label{prop:reductor-property-3}
Given $S \in \mathbb{S}$ and the propositional formula $\phi$, we have that ${\sigma_{\lnot\phi}(S)=S \setminus \sigma_{\phi}(S)}$.
\end{proposition}
\begin{proof}
By definition~\ref{def:selector}, $p \in \sigma_{\lnot\phi}(S) \iff {p \in S~\land~\lnot\phi(p)} \iff p \in S~\land~p \notin \sigma_{\phi}(S) \iff p \in S \setminus \sigma_{\phi}(S)$.
\end{proof}

\begin{proposition}\label{prop:reductor-property-4}
Given $S_1,S_2 \in \mathbb{S}$ and a propositional formula $\phi$, we have that $\sigma_\phi(S_1 \cup S_2)=\sigma_\phi(S_1) \cup \sigma_\phi(S_2)$.
\end{proposition}
\begin{proof}
By definition~\ref{def:selector}, $p \in \sigma_\phi(S_1 \cup S_2) \iff {p \in S_1 \cup S_2~\land~\phi(p)} \iff p \in S_1~\lor~p \in S_2~\land~\phi(p) \iff {p \in \sigma_\phi(S_1)~\lor~p \in \sigma_\phi(S_2)} \iff {p \in \sigma_\phi(S_1) \cup \sigma_\phi(S_2)}$.
\end{proof}

\begin{proposition}\label{prop:reductor-property-5}
Given $S_1,S_2 \in \mathbb{S}$ and a propositional formula $\phi$, we have that $\sigma_\phi(S_1 \cap S_2)=\sigma_\phi(S_1) \cap \sigma_\phi(S_2)$.
\end{proposition}
\begin{proof}
By definition~\ref{def:selector}, $p \in \sigma_\phi(S_1 \cap S_2) \iff {p \in S_1 \cap S_2~\land~\phi(p)} \iff p \in S_1~\land~p \in S_2~\land~\phi(p) \iff {p \in \sigma_\phi(S_1)~\land~p \in \sigma_\phi(S_2)} \iff {p \in \sigma_\phi(S_1) \cap \sigma_\phi(S_2)}$.
\end{proof}

\begin{proposition}\label{prop:reductor-property-6}
Given $S_1,S_2 \in \mathbb{S}$ and a propositional formula $\phi$, we have that $\sigma_\phi(S_1 \cap S_2)=\sigma_\phi(S_1) \cap S_2$.
\end{proposition}
\begin{proof}
By definition~\ref{def:selector}, $p \in \sigma_\phi(S_1 \cap S_2) \iff {p \in S_1 \cap S_2~\land~\phi(p)} \iff p \in S_1~\land~\phi(p)~\land~p \in S_2 \iff {p \in \sigma_\phi(S_1)~\land~p \in S_2} \iff p \in \sigma_\phi(S_1) \cap S_2$.
\end{proof}

\begin{proposition}\label{prop:reductor-property-7}
Given $S_1,S_2 \in \mathbb{S}$ and a propositional formula $\phi$, we have that $\sigma_\phi(S_1 \cap S_2)=S_1 \cap \sigma_\phi(S_2)$.
\end{proposition}
\begin{proof}
By definition~\ref{def:selector}, $p \in \sigma_\phi(S_1 \cap S_2) \iff {p \in S_1 \cap S_2~\land~\phi(p)} \iff p \in S_1~\land~p \in S_2~\land~\phi(p) \iff {p \in S_1~\land~p \in \sigma_\phi(S_2)} \iff p \in S_1 \cap \sigma_\phi(S_2)$.
\end{proof}

\begin{proposition}\label{prop:reductor-property-8}
Given $S \in \mathbb{S}$ and the propositional formulas $\phi$ and $\psi$, we have that $\sigma_{\phi}(\sigma_{\psi}(S))=\sigma_{\psi}(\sigma_{\phi}(S))$.
\end{proposition}
\begin{proof}
By definition~\ref{def:selector}, $p \in \sigma_{\phi}(\sigma_{\psi}(S)) \iff {p \in \sigma_{\psi}(S)~\land~\phi(p)} \iff p \in S~\land~\psi(p)~\land~\phi(p) \iff {p \in \sigma_{\phi}(S)~\land~\psi(p)} \iff p \in \sigma_{\psi}(\sigma_{\phi}(S))$.
\end{proof}

Fig.~\ref{fig:reductor-properties} summarises the properties of the reductor operator (i.e., Propositions~\ref{prop:reductor-property-1}-\ref{prop:reductor-property-8}).

\begin{figure}[!h]
\begin{center}
\begin{tabular}{ |m{18em} m{13em}| } 
\hline
${\sigma_{\phi\land\psi}(S)=\sigma_{\phi}(S) \cap \sigma_{\psi}(S)}$ & $\sigma_\phi(S_1 \cap S_2)=\sigma_\phi(S_1) \cap \sigma_\phi(S_2)$ \\
$\sigma_{\phi\lor\psi}(S)=\sigma_{\phi}(S) \cup \sigma_{\psi}(S)$ & $\sigma_\phi(S_1 \cap S_2)=\sigma_\phi(S_1) \cap S_2$ \\
${\sigma_{\lnot\phi}(S)=S \setminus \sigma_{\phi}(S)}$ & $\sigma_\phi(S_1 \cap S_2)=S_1 \cap \sigma_\phi(S_2)$ \\
$\sigma_\phi(S_1 \cup S_2)=\sigma_\phi(S_1) \cup \sigma_\phi(S_2)$ & $\sigma_{\phi}(\sigma_{\psi}(S))=\sigma_{\psi}(\sigma_{\phi}(S))$ \\
\hline
\end{tabular}
\end{center}
\caption{Reductor operator properties.}
\label{fig:reductor-properties}
\end{figure}

\section{Example}
\label{sec:example}

In this section, we provide an example for (inductively) composing a third-order space through the application of composition and reduction operations.

Assuming that $S_1,S_2,S_3,S_4,S_5 \in \mathbb{P}$, we start the composition process by defining the first-order parallel space $S_6=\PAR(S_1,S_2)$ which we reduce into ${A=\sigma_{p(S_1)=1~\land~p(S_2)=1}(S_6)}$. This process is illustrated in Fig.~\ref{fig:reduction-s6-a}.

\begin{figure}[!h]
\center
\begin{tikzpicture}[level distance=1cm,
  level 1/.style={sibling distance=1cm},
  level 2/.style={sibling distance=1cm}]
  \node(par1)[PAR,scale=1]{}
    child {node(s1)[draw] {$S_1$}}
	child {node(s2)[draw] {$S_2$}};

\node at ($(par1)+(-0.75,0.1)$) (s6) {$S_6$}; 
\draw[dashed]($(par1) + (-1,0.3)$)rectangle($(s2) + (0.5,-0.3)$);

\node at (3.5,0) (sel) {$A=\sigma_{p(S_1)=1~\land~p(S_2)=1}(S_6)$}; 
\draw[-{Triangle[width=10pt,length=8pt]}, line width=3pt](3,-0.5) -- (4, -0.5);

\node[circle,draw=black,inner sep=0pt,minimum size=5pt] at (7.2,0.1) (fork) {};
\node[draw] at (6.7,-0.8) (i1) {$S_1$};
\node[draw] at (7.7,-0.8) (i2) {$S_2$};
\node[circle,fill=black,inner sep=0pt,minimum size=5pt] at (7.2,-1.8) (join) {};

\node at ($(fork)+(-0.75,0)$) (a) {$A$}; 
\draw[dashed]($(i1) + (-0.5,1.3)$)rectangle($(i2) + (0.5,-1.2)$);

\draw[->] (fork) -- (i1); \draw[->] (i1) -- (join);
\draw[->] (fork) -- (i2); \draw[->] (i2) -- (join);
\end{tikzpicture}
\caption{Reducing ${S_6=\protect\PAR(S_1,S_2)}$ into ${A=\sigma_{p(S_1)=1~\land~p(S_2)=1}(S_6)=\{\{(S_1,1),(S_2,1)\}\}}$.}
\label{fig:reduction-s6-a}
\end{figure}

Likewise, we define the first-order sequential space ${S_7=\SEQ(S_3,S_4)}$ and reduce it into ${B=\sigma_{p(1)=S_3~\land~p(2)=S_4}(S_7)}$, as depicted in Fig.~\ref{fig:reduction-s7-b}.

\begin{figure*}[h]
\center
\begin{tikzpicture}[level distance=1cm,
  level 1/.style={sibling distance=1cm},
  level 2/.style={sibling distance=1cm}]
  \node(seq1)[SEQ,scale=1]{}
    child {node(s3)[draw] {$S_3$}}
	child {node(s4)[draw] {$S_4$}};

\node at ($(seq1)+(-0.75,0.1)$) (s7) {$S_7$}; 
\draw[dashed]($(seq1) + (-1,0.3)$)rectangle($(s4) + (0.5,-0.3)$);

\node at (3.5,0) (sel) {$B=\sigma_{p(1)=S_3~\land~p(2)=S_4}(S_7)$}; 
\draw[-{Triangle[width=10pt,length=8pt]}, line width=3pt](3,-0.5) -- (4, -0.5);

\node[draw] at (6.5,-0.8) (0) {$S_3$};
\node[draw] at (8,-0.8) (1) {$S_4$};

\node at ($(1)+(-1.75,0.7)$) (b) {$B$}; 
\draw[dashed]($(1) + (-2,0.9)$)rectangle($(1) + (0.5,-0.5)$);
\draw[->] (0) -- (1);
\end{tikzpicture}
\caption{Reducing ${S_7=\protect\SEQ(S_3,S_4)}$ into ${B=\sigma_{p(1)=S_3~\land~p(2)=S_4}(S_7)=\{\langle S_3,S_4 \rangle\}}$.}
\label{fig:reduction-s7-b}
\end{figure*}

In the next level of the composition structure, we define the second-order parallel space $S_8=\PAR(A,B)$ in terms of the above reductions, and we reduce it into ${C=\sigma_{p(A)=2~\land~p(B)=1}(S_8)}$. This is illustrated in Fig.~\ref{fig:reduction-s8-c}.\footnote{For clarity, we do not show the internal structure of the spaces $A$ and $B$.}

\begin{figure}[!h]
\center
\begin{tikzpicture}[level distance=1cm,
  level 1/.style={sibling distance=1cm},
  level 2/.style={sibling distance=1cm}]
  \node(par1)[PAR,scale=1]{}
    child {node(A)[draw,dashed] {$A$}}
	child {node(B)[draw,dashed] {$B$}};

\node at ($(par1)+(-0.75,0.1)$) (s8) {$S_8$}; 
\draw[dashed]($(par1) + (-1,0.3)$)rectangle($(B) + (0.5,-0.3)$);

\node at (3.5,0) (sel) {$C=\sigma_{p(A)=2~\land~p(B)=1}(S_8)$}; 
\draw[-{Triangle[width=10pt,length=8pt]}, line width=3pt](3,-0.5) -- (4, -0.5);

\node[circle,draw=black,inner sep=0pt,minimum size=5pt] at (7.7,0.1) (fork) {};
\node[draw,dashed] at (6.7,-0.8) (i1) {$A$};
\node[draw,dashed] at (7.7,-0.8) (i2) {$A$};
\node[draw,dashed] at (8.7,-0.8) (i3) {$B$};
\node[circle,fill=black,inner sep=0pt,minimum size=5pt] at (7.7,-1.8) (join) {};

\node at ($(fork)+(-1.25,0)$) (C) {$C$}; 
\draw[dashed]($(i1) + (-0.5,1.3)$)rectangle($(i3) + (0.5,-1.2)$);

\draw[->] (fork) -- (i1); \draw[->] (i1) -- (join);
\draw[->] (fork) -- (i2); \draw[->] (i2) -- (join);
\draw[->] (fork) -- (i3); \draw[->] (i3) -- (join);

\end{tikzpicture}
\caption{Reducing ${S_8=\protect\PAR(A,B)}$ into ${C=\sigma_{p(A)=2~\land~p(B)=1}(S_8)=\{\{(A,2),(B,1)\}\}}$.}
\label{fig:reduction-s8-c}
\end{figure}

Finally, we construct the third-order sequential space ${S_9=\SEQ(C,S_5)}$ and reduce it into $D=\sigma_{p(1)=C~\land~p(2)=S_5}(S_9)$. Diagramatically, we have the Fig.~\ref{fig:reduction-s9-d}.\footnote{For clarity, we do not show the internal structure of the space $C$.}

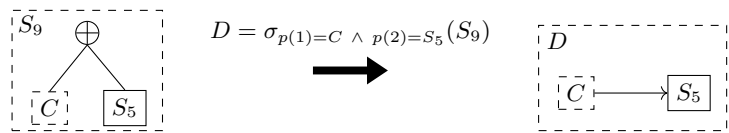
\begin{figure}[!h]
\center
\begin{tikzpicture}[level distance=1cm,
  level 1/.style={sibling distance=1cm},
  level 2/.style={sibling distance=1cm}]
  \node(seq1)[SEQ,scale=1]{}
    child {node(C)[draw,dashed] {$C$}}
	child {node(s5)[draw] {$S_5$}};

\node at ($(seq1)+(-0.75,0.1)$) (s9) {$S_9$}; 
\draw[dashed]($(seq1) + (-1,0.3)$)rectangle($(s4) + (0.5,-0.3)$);

\node at (3.5,0) (sel) {$D=\sigma_{p(1)=C~\land~p(2)=S_5}(S_9)$}; 
\draw[-{Triangle[width=10pt,length=8pt]}, line width=3pt](3,-0.5) -- (4, -0.5);

\node[draw,dashed] at (6.5,-0.8) (0) {$C$};
\node[draw] at (8,-0.8) (1) {$S_5$};

\node at ($(1)+(-1.75,0.7)$) (D) {$D$}; 
\draw[dashed]($(1) + (-2,0.9)$)rectangle($(1) + (0.5,-0.5)$);

\draw[->] (0) -- (1);
\end{tikzpicture}
\caption{Reducing ${S_9=\protect\SEQ(C,S_5)}$ into ${D=\sigma_{p(1)=C~\land~p(2)=S_5}(S_9)=\{\langle C,S_5 \rangle\}}$.}
\label{fig:reduction-s9-d}
\end{figure}

All the above operations can algebraically be expressed as shown in Fig.~\ref{fig:complete-example}.

\begin{figure}[!h]
\begin{center}
\begin{tabular}{ |m{10em} m{22em}| } 
\hline
$S_6 = \PAR(S_1,S_2)$ & $A = \sigma_{p(S_1)=1~\land~p(S_2)=1}(S_6) = \{\{(S_1,1),(S_2,1)\}\}$ \\
$S_7 = \SEQ(S_3,S_4)$ & $B = \sigma_{p(1)=S_3~\land~p(2)=S_4}(S_7) = \{\langle S_3,S_4 \rangle\}$ \\
$S_8 = \PAR(A,B)$ & $C = \sigma_{p(A)=2~\land~p(B)=1}(S_8) = \{\{(A,2),(B,1)\}\}$ \\
$S_9 = \SEQ(C,S_5)$ & $D = \sigma_{p(1)=C~\land~p(2)=S_5}(S_9) = \{\langle C,S_5 \rangle\}$ \\
\hline
\end{tabular} 
\end{center}
\caption{Operations to algebraically construct the third-order space $D$, where $S_j \in \mathbb{P}$ for all $1 \leq i \leq 5$ and $S_j \in \mathbb{C}$ for all $6 \leq j \leq 9$.}
\label{fig:complete-example}
\end{figure}

Although our example reduces spaces as composition is performed, it is important to note that it is also possible to define higher-order spaces without reducing to singleton sets. In fact, it is possible to define higher-order spaces that are never reduced, as described in Section~\ref{sec:semantics-higherorder}.

\section{Conclusions}
\label{sec:conclusions}

In this paper, we presented the semantics of an algebraic model for the inductive construction of computon spaces (i.e., sets of sequential and/or parallel constructs). Contrary to the traditional composition view, this model provides operators for the composition and for the reduction of computon spaces (not for individual programs). We limited ourselves to describe operators for sequencing, parallelisation and aggregation. In the future, we would like to investigate if it is possible to generate spaces of branchial and recursive computons. 

Furthermore, as program spaces are \emph{function spaces}, we plan to leverage the large body of theorems from the field of functional analysis. We also plan to provide computational interpretations for sequential and parallel computons. For example, a possible interpretation for $\langle S_1,S_2 \rangle$ is to simultaneously execute all the computons in $S_1$ and then all the computons in $S_2$. Similarly, a parallel computon $\{(S_1,1),(S_1,1)\}$ can be interpreted as the parallel execution of all the computons in both $S_1$ and $S_2$.

\bibliographystyle{splncs04}
\bibliography{refs}
\end{document}